\documentclass[conference]{IEEEtran}
\IEEEoverridecommandlockouts
\usepackage{footnote}
\usepackage[utf8]{inputenc}
\usepackage[english]{babel}
\usepackage[T1]{fontenc}
\usepackage{amsmath, amssymb, amsthm}
\usepackage{mathtools}
\usepackage{braket}
\usepackage{cuted}
\usepackage{halloweenmath}
\usepackage{cite}
\usepackage{graphicx}
\usepackage{xcolor}
\usepackage[dvipsnames]{xcolor}
\usepackage{tcolorbox}
\usepackage{tabularx}
\usepackage{colortbl}
\usepackage{tikz}
\usepackage{enumitem}
\usepackage{textcomp}
\usepackage{algorithm}
\usepackage{algpseudocode}
\usepackage{float}
\usepackage{placeins}
\usepackage[font=footnotesize]{caption}
\usepackage{subcaption}
\usepackage[left=0.64in,right=0.64in,top=0.72in,bottom=1.07in]{geometry}
\newtheorem{theorem}{Theorem}
\newtheorem{lemma}{Lemma}

\newtheorem{definition}{Definition}

\newlength{\myeqskip} 
\AtBeginDocument{%
    \setlength\abovedisplayskip{\myeqskip}%
    \setlength\belowdisplayskip{\myeqskip}%
    \setlength\abovedisplayshortskip{\myeqskip-\baselineskip}%
    \setlength\belowdisplayshortskip{\myeqskip}}
\def\BibTeX{{\rm B\kern-.05em{\sc i\kern-.025em b}\kern-.08em
		T\kern-.1667em\lower.7ex\hbox{E}\kern-.125emX}}

\allowdisplaybreaks
\begin{document}


\title{Two-dimensional Entanglement-assisted Quantum Quasi-cyclic Low-density Parity-check Codes}
 \author{\IEEEauthorblockN{Pavan Kumar and Shayan  Srinivasa Garani}
  \IEEEauthorblockA{Department of Electronic Systems Engineering, Indian Institute of Science, Bengaluru-560012, India\\
    Emails: \{pavankumar1, shayangs\}@iisc.ac.in }
    }
\maketitle
\begin{abstract}
For any positive integer $g \ge 2$, we derive general condition for the existence of a $2g$-cycle in the Tanner graph of two-dimensional ($2$-D) classical quasi-cyclic (QC) low-density parity-check (LDPC) codes. Depending on whether $p$ is an odd prime or a composite number, we construct two distinct families of $2$-D classical QC-LDPC codes with girth $>4$ by stacking $p \times p \times p$ tensors. Furthermore, using generalized Behrend sequences, we propose an additional family of $2$-D classical QC-LDPC codes with girth $>6$, constructed via a similar tensor-stacking approach.  All the proposed $2\text{-D}$ classical QC-LDPC codes exhibit an erasure correction capability of at least $p \times p$. Based on the constructed $2\text{-D}$ classical QC-LDPC codes, we derive two families of $2\text{-D}$  entanglement-assisted (EA) quantum low-density parity-check (QLDPC) codes. The first family of $2\text{-D}$  EA-QLDPC codes is obtained from a pair of $2\text{-D}$  classical QC-LDPC codes and is designed such that the unassisted part of the Tanner graph of the resulting EA-QLDPC code is free of $4$-cycles, while requiring only a single ebit to be shared across the quantum transceiver. The second family is constructed from a single $2\text{-D}$  classical QC-LDPC code whose Tanner graph is free from $4$-cycles. Moreover, the constructed EA-QLDPC codes inherit an erasure correction capability of $p \times p$, as the underlying classical codes possess the same erasure correction property.
\end{abstract}


\section{Introduction}
Following the success of one-dimensional (1-D) LDPC codes in classical communication and storage systems, higher-dimensional LDPC codes have also attracted significant research interest for the design of more powerful coding schemes. This growing interest is primarily driven by the stringent reliability requirements of advanced storage technologies such as shingled magnetic recording, two-dimensional (2-D) magnetic recording (TDMR)~\cite{wood2015two}, and three-dimensional NAND flash memories~\cite{kim2012three}, etc. In particular, a fundamental challenge in TDMR systems is that the dominant error mechanisms—including random errors, burst errors, and erasures—are inherently two-dimensional, arising from inter-track interference and inter-symbol interference.
 As a result, conventional $1$-D coding and decoding approaches are often inadequate and inefficient, motivating the development of native two-dimensional signal-processing and coding techniques attuned to the underlying channel structure~\cite{garani2018signal,chen2013joint}.

2-D burst errors and erasures can be handled efficiently using $2$-D LDPC codes, which naturally extend the sparse graphical structure of classical LDPC codes to higher dimensions. In this direction, several works have investigated the construction, encoding, and decoding of $2$-D LDPC codes, demonstrating their effectiveness in mitigating structured $2\text{-D}$ error patterns~\cite{mondal2021efficient2,kamabe2019burst,bharadwaj2025efficient}.  These developments motivate the design of novel 2-D classical LDPC codes as a foundation for constructing quantum LDPC codes capable of correcting $2\text{-D}$ errors and erasures under Markovian constraints, which are particularly relevant for quantum storage and quantum communication systems~\cite{fan2018construction,gu2024quantum}. Although some studies have explored 2-D quantum LDPC code constructions, important limitations remain. For example, the authors of~\cite{yang2025spatially} construct quantum LDPC codes from 2-D spatially coupled LDPC codes which are not tailored to mitigate 2-D burst errors and erasures. Other recent works, such as~\cite{berthusen2025toward} and~\cite{ruiz2025ldpc}, focus on quantum LDPC code design without explicitly addressing the correction of 2-D  burst error/erasure patterns. Consequently, existing approaches do not fully account for correlated error structures arising from Markovian effects. This gap motivates our work, in which we first construct novel 2-D classical QC-LDPC codes and then utilize them to develop 2-D quantum LDPC codes specifically designed to correct 2-D errors and erasures under realistic Markovian noise models.

In this work, we focus on the construction of $2\text{-D}$ entanglement-assisted (EA) quantum LDPC codes. We adopt the entanglement-assisted framework rather than the entanglement-unassisted setting to circumvent unavoidable cycles of length four that severely degrades the decoding performance. In other words, in the EA setting, the unassisted portion of the overall Tanner graph can be designed to be free of 4-cycles. Moreover, the entanglement-unassisted framework requires the underlying classical codes to satisfy the restrictive dual-containment constraint. In contrast, the availability of pre-shared ebits in the EA framework relaxes the stringent orthogonality requirements on the parity-check matrices, thereby allowing the construction of 2-D quantum LDPC codes from arbitrary classical codes. With the drive towards implementation of high-fidelity controlled quantum gates and caching of EPR pairs generated through a physical process, the EA construction of such 2-D QLDPC codes would be a promising direction. The main contributions of this work are summarized as follows.

\begin{enumerate}[itemsep=0pt,leftmargin=0pt]
\item[] (1) We first derive a general condition for the existence of $2g$-cycle in the Tanner graph of 2-D classical LDPC code. Depending on whether $p$ is an odd prime or a composite number, we construct two distinct families of $2$-D classical QC-LDPC codes with girth greater than $4$ by stacking $p \times p \times p$ tensors. Furthermore, using generalized Behrend sequences, we propose an additional family of $2$-D classical QC-LDPC codes with girth  greater than $6$, constructed via a similar tensor-stacking approach.  All the proposed code families achieve a burst erasure correction capability of at least $p \times p$.
\item[] (2) We then employ the constructed $2$-D classical QC-LDPC codes to develop two families of EA quantum LDPC codes. In the first family, two distinct $2$-D classical QC-LDPC codes are used to ensure that the unassisted portion of the overall Tanner graph is free of $4$-cycles, while requiring \textit{minimally} a single ebit. The second family is constructed using a single $2$-D classical QC-LDPC code, where the underlying classical code itself is free of $4$-cycles.
\end{enumerate}

This article is organized as follows. In Section~\ref{sec:classical constuction}, we first derive a general condition for the existence of $2g$-cycle in the Tanner graph of 2-D classical LDPC code. We then construct several families of 2-D classical QC-LDPC codes with burst erasure correction capabilities of at least $p\times p$. In Section~\ref{sec:quantumconstruction}, we construct two families of 2-D EA quantum LDPC codes: one obtained from a pair of distinct $2$-D classical LDPC codes and the other constructed using a single $2$-D classical LDPC code. Finally, Section~\ref{sec:conclusion} concludes the paper.

\section{2-D Classical LDPC Codes Constructions}\label{sec:classical constuction}
 Matcha et al.~\cite{matcha2018two} extended the principles of 1-D QC-LDPC code constructions to the 2-D setting, resulting in native 2-D LDPC codes capable of correcting 2-D burst errors/erasures. Following this work, several related constructions have been proposed, for example by Mondal and Garani~\cite{mondal2021efficient2} and Bharadwaj and Garani~\cite{bharadwaj2025efficient}. However, in all prior works, a general condition for the existence of $2g$-cycles is missing, and the constructions are restricted to girth $6$. To address this gap, we first briefly review the construction proposed in~\cite{matcha2018two}, and then describe the construction of the 2-D classical LDPC codes adopted in this work.
 
Before presenting the construction of 2-D classical LDPC codes, we introduce several definitions and notations that will be used throughout the paper.

\begin{figure}[htbp]
    \centering
    \begin{align*}
         \begin{subfigure}[t]{0.25\textwidth}
        \centering        \includegraphics[width=\linewidth]{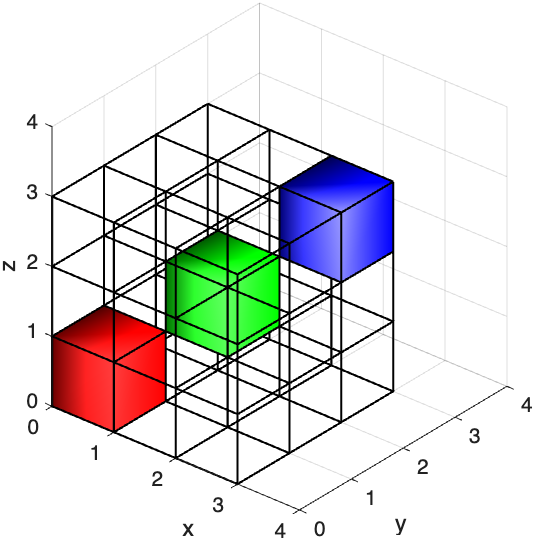}
        \caption{Identity Tensor $I_{3D}$.}
        \label{fig:sub1}
    \end{subfigure}&
    \begin{subfigure}[t]{0.25\textwidth}
        \centering        \includegraphics[width=\linewidth]{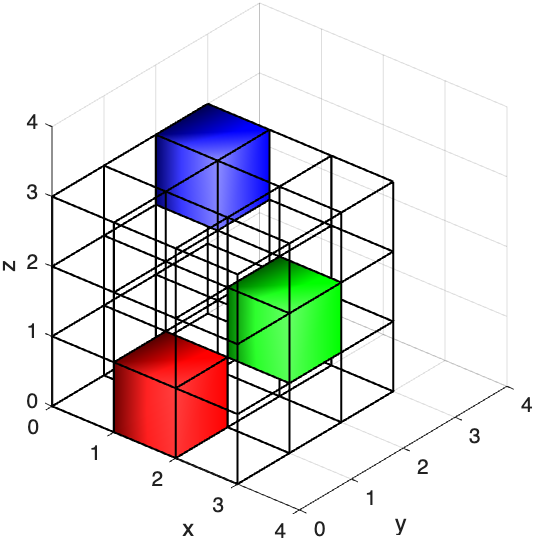}
        \caption{One $X$-shift, i.e., $P(I_{3D}).$}
        \label{fig:sub2}
    \end{subfigure}\\
    \bigskip\\
      \begin{subfigure}[t]{0.25\textwidth}
        \centering
        \includegraphics[width=\linewidth]{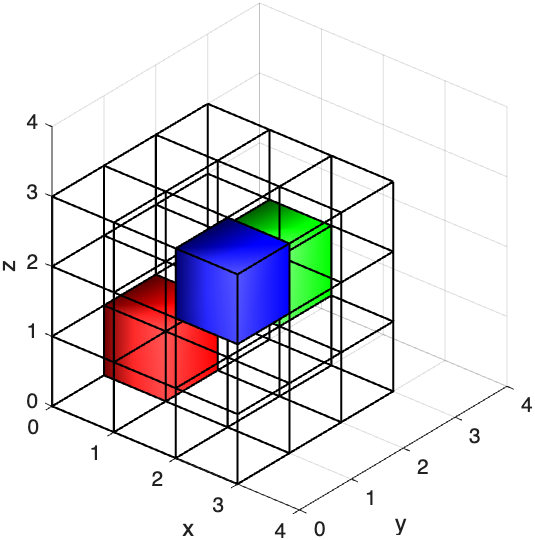}
        \caption{One $Y$-shift, i.e., $Q(I_{3D})$.}
        \label{fig:sub3}
    \end{subfigure}&
    \begin{subfigure}[t]{0.25\textwidth}
        \centering
        \includegraphics[width=\linewidth]{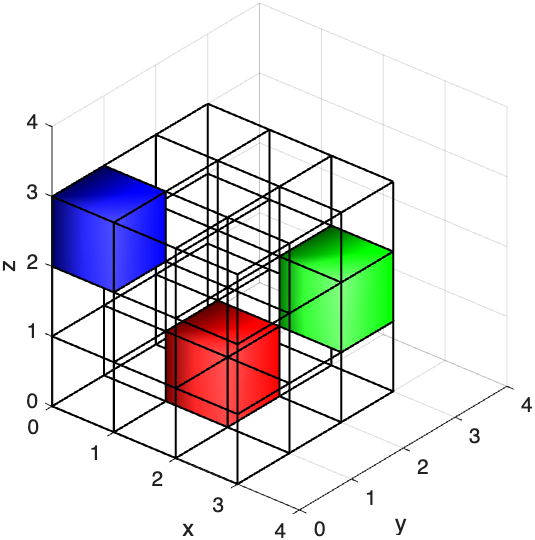}
        \caption{One $X$-shift and one $Y$-shift.}
        \label{fig:sub4}
    \end{subfigure}\\
    \end{align*}   
    \caption{Identity tensor and its various shifted versions, where the shifts are applied in the positive $i$- and $j$-directions.}
    \label{fig:identitytensor}
\end{figure}

Analogous to 1-D QC-LDPC array codes, where circulant permutation matrices are used to construct 1-D QC-LDPC codes~\cite{fossorier2004quasicyclic}, the construction of $2$-D classical LDPC codes employs an identity tensor, denoted by $I_{\text{3-D}}$, along with its shifted versions, where the identity tensor  $I_{\text{3-D}}$~(refer to Fig. \ref{fig:identitytensor}) of size $p \times p \times p$, where $p$ is a positive integer greater than~2, is defined as follows:
\begin{equation}\label{eq:identitytensor}
I_{3\text{-D}} = [I_{x,y,z}]_{x,y,z=0}^{p-1},\text{ where } I_{x,y,z} = \begin{cases}
1, & \text{if } x = y = z.\\
0, & \text{otherwise.}
\end{cases}
\end{equation}
 Let $\mathbb{Z}_{p}=\{0,1,2,\ldots,(p-1)\}$ be the ring of integers, for some integer $p>2$. We define the three permutation operators $P:\mathbb{Z}_{p} \to \mathbb{Z}_{p} $, $Q:\mathbb{Z}_{p} \to \mathbb{Z}_{p} $ and $R:\mathbb{Z}_{p} \to \mathbb{Z}_{p} $ as follows:
\begin{equation}
P(i) = Q(i) = R(i) = \begin{cases}
p-1, & i = 0, \\
i - 1, & \text{otherwise.}
\end{cases}
\end{equation}
Now, we apply the permutation operators  $P$, $Q$ and $R$ on the identity tensor $I_{3\text{-D}} = [I_{x,y,z}]_{x,y,z=0}^{p-1}$, as follows:
\begin{equation}\label{eq:permuationoperators}
\left.
\begin{aligned}
P(I_{3\text{-D}}) &= [I_{P(x),y,z}]_{x,y,z=0}^{p-1}, \\
Q(I_{3\text{-D}}) &= [I_{x, Q(y),z}]_{x,y,z=0}^{p-1}, \\
R(I_{3\text{-D}}) &= [I_{x,y,R(z)}]_{x,y,z=0}^{p-1}.
\end{aligned}
\right\}
\end{equation}
From~\eqref{eq:permuationoperators}, it is evident that the application of $P$, $Q$ and $R$ to the identity tensor $I_{3\text{-D}}$ results in cyclic shifts along the positive $x$-, $y$-, and $z$-axes, respectively. Moreover, we refer to the shifts along the $x$ -, $y$-, and $z$-axes as the $X$-, $Y$-, and $Z$-shifts, respectively. These shifts have been illustrated for the identity tensor $I_{3\text{-D}}$ of size $3 \times 3 \times 3$ in Fig.~\ref{fig:identitytensor}.

We now establish an important relation among the permutation operators $P$, $Q$ and $R$ acting on the identity tensor $I_{3\text{-D}}$, which is stated in the following lemma.

\begin{lemma}\label{le:permutationrelation} For $c\in\mathbb{Z}_{p}$, let $P$, $Q$ and $R$ be the permutation operators defined in~\eqref{eq:permuationoperators}. Then the following relations hold.
\begin{enumerate}
    \item $P^{c} \circ Q^{c}\bigl(I_{3\text{-D}}\bigr) = R^{\,p-c}\bigl(I_{3\text{-D}}\bigr).$
    \item $Q^{c} \circ R^{c}\bigl(I_{3\text{-D}}\bigr) = P^{\,p-c}\bigl(I_{3\text{-D}}\bigr).$
    \item $R^{c} \circ P^{c}\bigl(I_{3\text{-D}}\bigr) = Q^{\,p-c}\bigl(I_{3\text{-D}}\bigr).$
\end{enumerate}
\end{lemma}
\begin{proof} We prove the first statement; the remaining statements follow analogously. To prove that $P^{c}\circ Q^{c}\bigl(I_{3\text{-D}}\bigr) = R^{\,p-c}\bigl(I_{3\text{-D}}\bigr),$ it is sufficient to verify that the two tensors coincide entrywise.

Applying the permutation operator $Q^{c}$ shifts the identity tensor along the $y$-axis by $c$, and subsequently applying $P^{c}$ shifts it along the $x$-axis by $c$. Hence,
\begin{equation*}
    P^{c} \circ Q^{c}(I_{3\text{-D}})_{x,y,z}
= I_{x-c,\,y-c,\,z}.
\end{equation*}
Therefore, by equation \eqref{eq:identitytensor}, we have
\begin{equation*}
    Q^{c} \circ P^{c}(I_{3\text{-D}})_{x,y,z}=1
\text{ if and only if } x-c=y-c=z.
\end{equation*}
This condition is equivalent to $x=y=z+c$ in $\mathbb{Z}_{p}$.

On the other hand, applying the operator $R^{p-c}$ corresponds to a shift of $p-c=-c \pmod{p}$ along the $z$-axis, and thus
\begin{equation*}
    R^{p-c}(I_{3\text{-D}})_{x,y,z}
= I_{x,\,y,\,z+c}.
\end{equation*}
Consequently, $R^{p-c}(I_{3\text{-D}})_{x,y,z}=1\iff$  $x=y=z+c.$
Since both tensors are equal to \(1\) under the same condition and zero otherwise, we conclude that $Q^{a} \circ P^{c}\bigl(I_{3\text{-D}}\bigr) = R^{p-c}\bigl(I_{3\text{-D}}\bigr).$ This completes the proof. \end{proof}
 Next, we show that any two of the permutation operators $P$, $Q$  and $R$ are sufficient to generate the action of all three, as stated in the following lemma.
 \begin{lemma}\label{le:permutationrelation2}
 For $a,b,c\in\mathbb{Z}_{p}$, let   $P$, $Q$ and $R$ be the permutation operators defined on the identity tensor $I_{3\text{-D}}$. Then we have
 \begin{enumerate}
     \item  $P^{a} \circ Q^{b}\circ R^{c}  (I_{3\text{-D}})= P^{(p+a-c)} \circ Q^{(p+b-c)} (I_{3\text{-D}}),$
    \item  $P^{a} \circ Q^{b}\circ R^{c}  (I_{3\text{-D}})= Q^{(p+b-a)} \circ R^{(p+c-a)} (I_{3\text{-D}}),$
    \item  $P^{a} \circ Q^{b}\circ R^{c}  (I_{3\text{-D}})= R^{(p+c-b)} \circ P^{(p+a-b)} (I_{3\text{-D}}).$
 \end{enumerate}  
 \end{lemma}
 \begin{proof} We provide the proof of the first statement; the remaining statements can be proved in a similar manner. By applying Lemma~\ref{le:permutationrelation} together with the commutative property of the permutation operators, we obtain
    \begin{align*}
P^{a} \circ Q^{b}\circ R^{c}  (I_{3\text{-D}}) &= P^{a} \circ Q^{b}\circ R^{p-(p-c)}  (I_{3\text{-D}}) \\
&= P^{a} \circ Q^{b} \circ P^{(p-c)} \circ Q^{(p-c)} (I_{3\text{-D}})\\
&= P^{(p+a-c)} \circ Q^{(p+b-c)} (I_{3\text{-D}}).
\end{align*} This completes the proof.
 \end{proof}
In light of Lemma~\ref{le:permutationrelation2}, we construct 2-D classical LDPC codes by choosing permutations \( P \) and \( Q \) that implement circular shifts along the positive  \( x \)- and \( y \)-axes, respectively. Based on this construction, we introduce the following assumptions and terminology to define the parity-check tensor \( H_{2\text{-D}} \), which is obtained as a stack of tensors corresponding to the 2-D classical LDPC codes.

\begin{enumerate}[leftmargin=*]
\item  We introduce a new set of coordinate axes, denoted by $i$, $j$ and $k$, aligned with the $x$-, $y$-, and $z$-directions, respectively. The indices $(x,y,z)$ are used to label the positions of individual small cubes within the tensor, while the indices $(i,j,k)$ refer to the  coordinates of a tensor as a whole. This distinction allows us to clearly differentiate between local indexing within the tensor and global indexing of the tensor in the parity-check tensor $H_{2\text{-D}}$.
    \item For given size of identity tensor, i.e., $p\times p\times p$, we first define the X- and Y-shifts of a tensor located at an arbitrary position $(i,j,k)$ along the $x$- and $y$-directions. The X- and Y-shifts are denoted by $a(i,j,k)$ and $b(i,j,k)$, respectively.
    \item For given positive integers $c$, $b$ and $h$ such that $p\mid h$, tensors with the prescribed shifts are stacked along the $i$-, $j$-, and $k$-directions. Within this stack, the $(i,j,k)$-th tensor, with indices satisfying $0 \leq i \leq c$, $0 \leq j \leq b$, and $0 \leq k \leq h$, is constructed as
    \begin{equation*}
        P^{a(i,j,k)} \circ Q^{b(i,j,k)}\bigl(I_{3\text{-D}}\bigr).
    \end{equation*}
    \item A layer of small cubes within the tensors, taken along the $xy$-plane, is referred to as a horizontal layer (refer to Fig.~\ref{fig:layer}) and will be denoted by $\ell$.
     \item A layer of small cubes within the tensors, taken along the $xz$-plane, is referred to as a vertical layer and will be denoted by $v$.
    \item During this stacking process, the first $p^{2}$ horizontal layers are arranged such that the colored small cubes cover all positions in the $xy$-plane, for $0\leq x\leq cp-1$ and $0\leq y\leq bp-1$. Each subsequent group of $p^{2}$ horizontal layers satisfies the same coverage property, and this pattern continues throughout the stacking.
    \item 
 Two horizontal layers $\ell_1$ and $\ell_2$ of $H_{2\text{-D}}$ are said to belong to the same horizontal block-layer if $\left\lfloor \frac{\ell_1}{p^{2}} \right\rfloor = \left\lfloor \frac{\ell_2}{p^{2}} \right\rfloor .$  
 \item Tensors sharing the same $i$ and $j$ indices are said to belong to the same block-column.
\item A closed path of length $2g$ in the parity-check tensor $H_{2\text{-D}}$ is defined as a sequence of index triplets $(\,\cdot,\cdot,\cdot\,)$,
where the first two coordinates correspond to the  block-column indices and the third coordinate corresponds to the horizontal block-layer index, arranged so as to form a closed alternating path of length $2g$ as follows:
\begin{equation*}
(i_{0},j_{0},k_{0}),(i_{1},j_{1},k_{0});
(i_{1},j_{1},k_{1}),(i_{2},j_{2},k_{1});\ldots,
\end{equation*}
\begin{equation*}
(i_{g-2},j_{g-2},k_{g-2}),(i_{g-1},j_{g-1},k_{g-2});
\end{equation*}
\begin{equation*}
(i_{g-1},j_{g-1},k_{g-1}),(i_{0},j_{0},k_{g-1}),
\end{equation*}
where $(i_\ell, j_\ell) \neq (i_{\ell+1}, j_{\ell+1})$ and $k_\ell \neq k_{\ell+1}$ for $\ell = 0,1,\ldots,g-2$, and
$(i_{g-1},j_{g-1}) \neq (i_{0},j_{0})$ with $k_{g-1} \neq k_{0}$.
\end{enumerate} 
To analyze the existence of cycles in the Tanner graph of $H_{2\text{-D}}$, we  unfold the 3-D parity-check tensor $H_{2\text{-D}}$ into a $2\text{-D}$ parity-check matrix $H_{1\text{-D}}$ of size $ph \times bcp^{2}$. The unfolding is performed layer by layer along the vertical dimension: the first $cp$ columns of $H_{1\text{-D}}$ correspond to the first vertical layer of $H_{2\text{-D}}$, the next $cp$ columns correspond to the second vertical layer, and so on. Mathematically, the parity-check tensor $H_{2\text{-D}}$ is unfolded into the matrix $H_{1\text{-D}}$ such that the $(i,j)$-th entry of $H_{1\text{-D}}$ is the $(j \pmod{p},\left\lfloor\frac{j}{p}\right\rfloor,i)$-th entry in $H_{2\text{-D}}$.

Now, consider a closed path of length $2g$ in the parity-check tensor $H_{2\text{-D}}$, as described under assumption~(9). For this closed path, let $\{\bar{k}_i\}_{i=0}^{g-1}$ denote the collection of $k$-indices, where each $\bar{k}_i$ represents the $k$-coordinate of a tensor located in the $k_i$-th horizontal block-layer and satisfies the conditions required to connect the corresponding cycle segments. Under these conditions, we arrive at the following theorem, which establishes the existence of cycles of length $2g$ in the Tanner graph associated with $H_{2\text{-D}}$ along the specified closed path.
\begin{theorem}[\textbf{Condition for a $2g$-cycle}]
\label{thm:2D-2k-cycle}
Let $H_{2\text{-D}}$ be the parity-check tensor composed of permutation tensors of size $p\times p\times p$. Assume that the permutation tensor located at position $(i,j,k)$ is specified
by the shift functions $a(i,j,k)$ and $b(i,j,k)$. Consider the arbitrary closed path of length $2g$ as described under assumption~(9), then the Tanner graph of $H_{2\text{-D}}$ contains a cycle of length $2g$ along the considered closed path if and only if there exists a collection of $k$-indices $\{\Bar{k}_{i}\}_{i=0}^{g-1}$, where each $\bar{k}_i$ is the $k$-coordinate of a tensor located in the $k_i\text{-th}$ horizontal block-layer, satisfying the conditions required to connect the corresponding cycle segments, and
\begin{align*}
\sum_{m=0}^{g-1}
\Bigl(
a(i_{m},j_{m},\Bar{k}_{m})
-
a(i_{m+1},j_{m+1},\Bar{k}_{m})
\Bigr)
&= 0 \pmod p,\\
\sum_{m=0}^{g-1}
\Bigl(
b(i_{m},j_{m},\Bar{k}_{m})
-
b(i_{m+1},j_{m+1},\Bar{k}_{m})
\Bigr)
&=0 \pmod p, 
\end{align*}
where $(i_g,j_g)=(i_0,j_0)$.
\end{theorem}
\begin{proof} Throughout the proof, $(i,j,k)$ denote the indices of the tensor, whereas $(x,y,z)$ denote the local indices of the small cubes, i.e., cubes within the tensor. Consider the arbitrary closed path of length $2g$ as follows:
    \begin{equation*}
(i_{0},j_{0},k_{0}),(i_{1},j_{1},k_{0});
(i_{1},j_{1},k_{1}),(i_{2},j_{2},k_{1});\ldots,
\end{equation*}
\begin{equation*}
(i_{g-2},j_{g-2},k_{g-2}),(i_{g-1},j_{g-1},k_{g-2});
\end{equation*}
\begin{equation*}
(i_{g-1},j_{g-1},k_{g-1}),(i_{0},j_{0},k_{g-1}),
\end{equation*}
where $(i_\ell, j_\ell) \neq (i_{\ell+1}, j_{\ell+1})$ and $k_\ell \neq k_{\ell+1}$ for $\ell = 0,1,2,\ldots,g-2$, and
$(i_{g-1},j_{g-1}) \neq (i_{0},j_{0})$ with $k_{g-1} \neq k_{0}$. Next, consider the following sequence of local indices corresponding to the nonzero entries along the given closed path:
 \begin{equation*}
(x_{0},y_{0},z_{0}),(x_{1},y_{1},z_{0});
(x_{1},y_{1},z_{1}),(x_{2},y_{2},z_{1});\ldots,
\end{equation*}
\begin{equation*}
(x_{g-2},y_{g-2},z_{g-2}),(x_{g-1},y_{g-1},z_{g-2});
\end{equation*}
\begin{equation*}
(x_{g-1},y_{g-1},z_{g-1}),(x_{g},y_{g},z_{g-1}),
\end{equation*}
where $0 \le x_i, y_i, z_i \le p-1$ for all $i, j, k$, since without loss of generality these indices are taken to lie within the tensor, which is sufficient to prove our result. It is straightforward to observe that, for any tensor located at position $(i,j,k)$, the $x$  coordinate of a nonzero entry in the $z$-th layer of the tensor, where $0 \le z \le p-1$, is given by
\begin{equation}\label{eq:generalicase}
    x=z+a(i,j,k)\pmod{p}.
\end{equation}
Similarly, for any tensor located at position $(i,j,k)$, the $y$-coordinate of a nonzero entry in the $z$-th layer of the tensor, where $0 \le z \le p-1$, is given by
\begin{equation}\label{eq:generaljcase}
    y=z+b(i,j,k)\pmod{p}.
\end{equation}
Thus, for the existence of a $2g$-cycle corresponding to the given closed path, it suffices to show that $(x_g,y_g) = (x_0,y_0)$. Here, $(x_g,y_g)$ denotes the $(x,y)$-coordinates of a nonzero entry in the layer $z = z_{g-1}$ of the tensor at position $(i_0,j_0,\Bar{k}_{g-1})$, where $\Bar{k}_{g-1}$ is the $k$-th coordinate of some tensor in the $k_{g-1}$ horizontal block-layer while $(x_0,y_0)$ denotes the $(x,y)$-coordinates of a nonzero entry in the layer $z = z_{0}$ of the tensor corresponding to $(i_0,j_0,\Bar{k}_{0})$, where $\Bar{k}_{0}$ is the $k$-th coordinate of some tensor in the $k_{0}$ horizontal block-layer. Now, by equation  \eqref{eq:generalicase}, for each index triplet along the closed path, the $x$-coordinates of the nonzero entries are given by the following set of
equations: 
\begin{align*}
   x_{0}&=z_{0}+a(i_{0},j_{0},\Bar{k}_{0})\pmod{p},\\ 
   x_{1}&=z_{0}+a(i_{1},j_{1},\Bar{k}_{0})\pmod{p},\\ 
    x_{1}&=z_{1}+a(i_{1},j_{1},\Bar{k}_{1})\pmod{p},\\ 
   x_{2}&=z_{1}+a(i_{2},j_{2},\Bar{k}_{1})\pmod{p},\\ 
   \vdots&\qquad\qquad\vdots\\
   x_{g-1}&=z_{g-1}+a(i_{g-1},j_{g-1},\Bar{k}_{g-1})\pmod{p},\\ 
   x_{g}&=z_{g-1}+a(i_{0},j_{0},\Bar{k}_{g-1})\pmod{p},
\end{align*}
where each $\bar{k}_i$ is the $k$-coordinate of a tensor located in the $k_i\text{-th}$ horizontal block-layer, satisfying the conditions required to connect the corresponding cycle segments.
By the above set of equations, we have the following equation under modulo $p$:
\begin{equation}\label{eq:icase}
    x_{g}=x_{0}+\sum_{m=0}^{g-1}
\Bigl(
a(i_{m},j_{m},\Bar{k}_{m})
-
a(i_{m+1},j_{m+1},\Bar{k}_{m})
\Bigr),
\end{equation}
where $(i_{g},j_{g})=(i_{0},j_{0})$.
 Now, by applying the same reasoning to the $y$-coordinates of the nonzero entries, we obtain the following equation under modulo ${p}$:
\begin{equation}\label{eq:jcase}
    y_{g}=y_{0}+\sum_{m=0}^{g-1}
\Bigl(
b(i_{m},j_{m},\Bar{k}_{m})
-
b(i_{m+1},j_{m+1},\Bar{k}_{m})
\Bigr),
\end{equation}
where $(i_{g},j_{g})=(i_{0},j_{0})$.
For the existence of a $2g$-cycle, it is necessary that $(x_{g},y_{g})=(x_{0},y_{0})$. 
By equations \eqref{eq:icase} and \eqref{eq:jcase}, this occurs if and only if the following two equations are simultaneously satisfied:
\begin{align*}
\sum_{m=0}^{g-1}
\Bigl(
a(i_{m},j_{m},\Bar{k}_{m})
-
a(i_{m+1},j_{m+1},\Bar{k}_{m})
\Bigr)
&= 0 \pmod p,\\
\sum_{m=0}^{g-1}
\Bigl(
b(i_{m},j_{m},\Bar{k}_{m})
-
b(i_{m+1},j_{m+1},\Bar{k}_{m})
\Bigr)
&=0 \pmod p, 
\end{align*}
where $(i_g,j_g)=(i_0,j_0)$. 
This completes the proof.
\end{proof}
\begin{figure}[t]
    \centering    \includegraphics[width=0.85\linewidth]{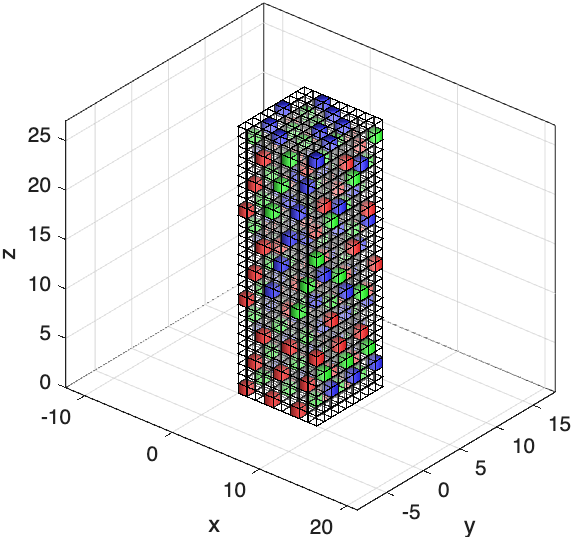}
    \caption{Parity-check tensor $H_{2\text{-D}}$ for $p=3$, where the shifts used are the same as those defined in~\eqref{shifts} with $\phi(i)=i$, and $\psi(j)=j$ and $\eta(k)=k$.}
    \label{fig:stack}
\end{figure}
Henceforth, this section is divided into two subsections. The first subsection focuses on the construction of 2-D classical QC-LDPC codes with girth greater than $4$, while the second subsection addresses the construction of 2-D classical QC-LDPC codes with girth greater than $6$.
\subsection{2-D Classical QC-LDPC Codes with Girth Greater than~$4$}
Depending on whether $p$ is an odd prime or composite number, in this subsection, we construct two families of $2\text{-D}$ classical QC-LDPC codes by stacking permutation tensors of size $p \times p \times p$.

Let $p$ be an odd prime, and this assumption applies throughout this subsection unless stated otherwise. We construct the $2$-D QC-LDPC parity-check tensor $H_{2\text{-D}}$ by stacking $p \times p \times p^{2}$ permutation  tensors, each of size $p \times p \times p$, along the $i$-, $j$-, and $k$-directions, respectively. These tensors are generated by applying various combinations of the shift operators $P$ and $Q$ to the identity tensor $I_{3\text{-D}}$, corresponding to $X$-shift and $Y$-shifts, respectively. While constructing $H_{2\text{-D}}$, the $(i, j, k)$th permutation tensor is given by $P^{a(i,j,k)} \circ Q^{b(i,j,k)} (I_{3\text{-D}})$, where $a(i, j, k)$ and $b(i, j, k)$ are chosen as follows:
\begin{equation}\label{shifts}
   \left.
\begin{aligned}
a(i, j, k) = \bmod\left(\eta(k) + \left\lfloor \frac{k}{p} \right\rfloor (\phi(i) + \psi(j)), p\right), \\
b(i, j, k)= \bmod\left(\left \lfloor \frac{k}{p} \right\rfloor \phi(i), p\right),
\end{aligned}
\right\} 
\end{equation}
where $\phi,\psi:\mathbb{Z}_{p}\to\mathbb{Z}_{p}$ are bijective mappings, and the mapping $\eta:\mathbb{Z}_{p^{2}}\to\mathbb{Z}_{p}$ is  defined as follows:
\begin{equation*}
    \eta(k)=\begin{cases}
        \zeta_{1}(k), &\text{ for } 0\leq k\leq (p-1);\\
         \zeta_{2}(\bmod{(k,p)}), &\text{ for } p\leq k\leq (2p-1);\\
           \qquad \vdots &\qquad\vdots\\
           \zeta_{p}(\bmod{(k,p)}), &\text{ for } p\leq k\leq (p^{2}-1);
    \end{cases}
\end{equation*}
where $\zeta_{i}:\mathbb{Z}_{p}\to\mathbb{Z}_{p}$ is a bijective mapping for each $i$.
\begin{figure}
    \centering
    \includegraphics[width=0.65\linewidth]{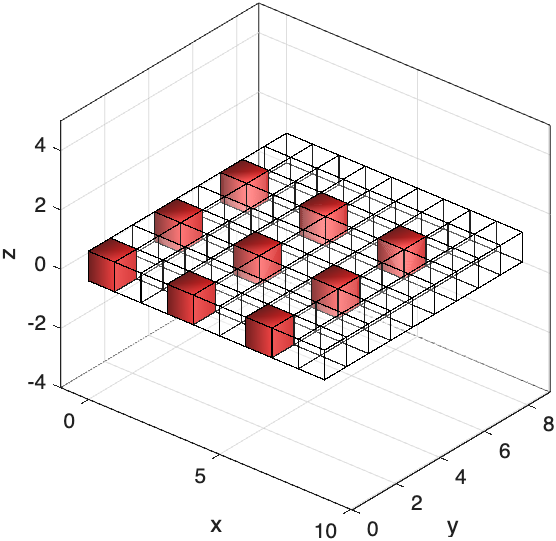}
    \caption{First layer in parity-check tensor $H_{2\text{-D}}$ for $p=3$, constructed using the specific shifts with $\phi(i)=i$, $\psi(j)=j$, and $\eta(k)=k \pmod{p}$ for all $i,j,k$.}
    \label{fig:layer}
\end{figure}

 An example of the parity-check tensor $H_{2\text{-D}}$ for \(p=3\), constructed using the specific shifts with $\phi(i)=i$, $\psi(j)=j$, and $\eta(k)=k \pmod{p}$ for all $i,j,k$, is shown in Fig.~\ref{fig:stack}. 
 
 Let \(H_{2\text{-D}}\) denote the parity-check tensor constructed using the shifts defined in~\eqref{shifts}. Then the corresponding Tanner graph is free of \(4\)-cycles as stated in the following lemma.
\begin{lemma}\label{le:4-cyclefree}
    Let $H_{2\text{-D}}$ be the parity-check tensor stack of $p \times p \times p^{2}$ permutation  tensors, each of size $p \times p \times p$, along the $i$-, $j$-, and $k$-directions, respectively, constructed with the shifts defined in \eqref{shifts}. Then the Tanner graph associated with $H_{2\text{-D}}$ is free from $4$-cycles.
\end{lemma}
\begin{proof}
    Consider arbitrary horizontal block-layer $k_{0}$ and $k_{1}$, where $0\leq k_{0}, k_{1}\leq (p-1)$. To check the existence of cycle of length $4$, consider the following arbitrary closed path of tensors indices:
     \begin{equation*}
(i_{0},j_{0},k_{0}),(i_{1},j_{1},k_{0});
(i_{1},j_{1},k_{1}),(i_{0},j_{0},k_{1}),
\end{equation*}
where $i_s, j_t \in \{0,1,2,\ldots,(p-1)\}$ for all $s,t \in \{0,1\}$.
 Then, by Theorem~\ref{thm:2D-2k-cycle}, a $4$-cycle exists along the closed path
$(i_{0}, j_{0}, k_{0}), (i_{1}, j_{1}, k_{0}), (i_{1}, j_{1}, k_{1}), (i_{0}, j_{0}, k_{1})$
if and only if there exists a set of $k$-indices $\{\bar{k}_{0}, \bar{k}_{1}\}$, where each $\bar{k}_i$ is the $k$-coordinate of a tensor located in the $k_i\text{-th}$ horizontal block-layer, satisfying the conditions required to connect the corresponding cycle segments, and the following equations are simultaneously satisfied under modulo $p$:
\begin{equation*}
    a(i_{0},j_{0},\Bar{k}_{1})-a(i_{0},j_{0},\Bar{k}_{0})= a(i_{1},j_{1},\Bar{k}_{1})-a(i_{1},j_{1},\Bar{k}_{0}),
\end{equation*}
\begin{equation*}
    b(i_{0},j_{0},\Bar{k}_{1})-b(i_{0},j_{0},\Bar{k}_{0})= b(i_{1},j_{1},\Bar{k}_{1})-b(i_{1},j_{1},\Bar{k}_{0}).
\end{equation*}
Now, by substituting the shifts defined in \eqref{shifts} into the above equations, a $4$-cycle exists in the Tanner graph of $H_{2\text{-D}}$ if and only if the following equations hold:
\begin{equation*}
   \alpha((\phi(i_0)-\phi(i_{1}))+( \psi(j_0)-\psi(j_{1})))=0\pmod{p},
\end{equation*}
\begin{equation*}
   \alpha((\phi(i_0)-\phi(i_{1}))=0\pmod{p},
\end{equation*}
where $ \alpha=\left(\left\lfloor \frac{\Bar{k}_{1}}{p}\right\rfloor-\left\lfloor \frac{\Bar{k}_{0}}{p}\right\rfloor\right)\neq 0$ $\pmod{p}$. Therefore, for the existence of $4$-cycle we, must have
\begin{equation}\label{eq:4-cyleeq1}
   ((\phi(i_0)-\phi(i_{1}))+( \psi(j_0)-\psi(j_{1})))=0\pmod{p}.
\end{equation}
\begin{equation}\label{eq:4-cyleeq2}
   ((\phi(i_0)-\phi(i_{1}))=0\pmod{p}.
\end{equation}
The equations~\eqref{eq:4-cyleeq1} and \eqref{eq:4-cyleeq2} cannot be satisfied simultaneously as $\phi$ and $\psi$ are bijective and $(i_{0},j_{0})\neq (i_{1},j_{1})$. This completes the proof.
\end{proof}
After establishing that the proposed parity-check tensor is free of $4$-cycles—and consequently that the associated 2-D code is also free of $4$-cycles—we next determine the exact rank of the proposed parity-check tensor, which in turn yields the exact code rate.

To this end, we adopt a geometric interpretation of the tensor entries: colored cubes correspond to entries equal to $1$, while transparent cubes represent entries equal to $0$. Let $\ell_{1}$ and $\ell_{2}$ denote any two horizontal layers of the parity-check tensor $H_{2\text{-D}}$. We then define the inner product between these layers as follows:
\begin{equation}\label{eq:innerproduct}
\langle \ell_{1}, \ell_{2}\rangle= \sum_{x,y=0}^{p^{2}-1}\ell^{(x,y)}_{1} \ell^{(x,y)}_{2},   
\end{equation}
where $ \ell^{(x,y)}_{i}$ represents the $(x,y)$th entry in the layer $\ell_{i}$. Now, we state the following lemma, which characterizes the inner product between any two layers of $H_{2\text{-D}}$ constructed using the shifts defined in \eqref{shifts}.
\begin{figure}
    \centering    \includegraphics[width=0.9\linewidth]{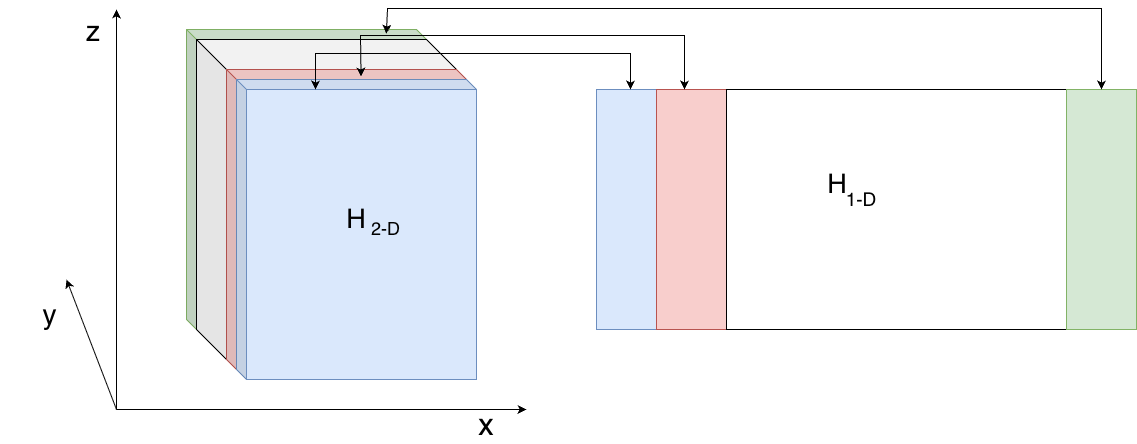}
   \caption{Unfolding of the $3$-D parity-check tensor $H_{2\text{-D}}$ into a $2$-D parity-check matrix $H_{1\text{-D}}$, where different colors represent a vertical layer in $H_{2\text{-D}}$.}
    \label{fig:2-dparitycheckmatrix}
\end{figure}
\begin{lemma}\label{maintheorem1} Let $H_{2\text{-D}}$ be the parity-check tensor  constructed based on the shifts in \eqref{shifts}, and
let $\ell_{1}$ and $\ell_{2}$ be any two horizontal layers in $H_{2\text{-D}}$, then we obtain
\begin{equation*}
  \langle \ell_{1}, \ell_{2}\rangle=\begin{cases}
    p^{2},&\text{if } \ell_{1}=\ell_{2}, \\
    1,&\text{if } \ell_{1}\neq\ell_{2} \text{ belong to a different}\\ &\text {horizontal block-layer},\\
    0,&\text{otherwise.} 
  \end{cases}  
\end{equation*}
\end{lemma}

\begin{proof} By taking into account the shifts defined in \eqref{shifts}, it follows directly that for the parity-check tensor $H_{2\text{-D}}$, $\langle \ell_{1}, \ell_{2} \rangle = p^{2} \quad \text{if } \ell_{1}=\ell_{2},$
and
$\langle \ell_{1}, \ell_{2} \rangle = 0$ if  $\ell_{1}\neq \ell_{2}$ and  $\ell_{1}, \ell_{2}$ belong to the same horizontal block-layer. It remains to show that the inner product between the layers $\ell_{1}$ and $\ell_{2}$ is one whenever $\ell_{1} \neq \ell_{2}$ and they belong to different horizontal block-layer.

To prove the remaining claim, we first determine the positions of nonzero entries that are common to the layers \( \ell_{1} \) and \( \ell_{2} \), and then use this information to compute the inner product. For this purpose, we begin by identifying the positions of the nonzero entries in an arbitrary layer of \( H_{2\text{-D}} \). In \( H_{2\text{-D}} \), for \( k = 0 \), both the \( I \)-shifts and \( J \)-shifts are zero for each permutation tensor. Consequently, the positions of the ones in the top layer of the parity-check tensor \( H_{2\text{-D}} \) are given by
\begin{equation}
S_1 = \{(ip, jp\} \mid 0 \le i \le (p-1), 0 \le j \le (p-1)\}.
\end{equation}
For an arbitrary layer $0\leq \ell \leq (p-1)$ in the bottom tensor tile corresponding to $k = 0$, i.e., for $\ell  = 0$ to $\ell = p-1$, such positions are given by
\begin{equation}
\{(ip + \ell,jp + \ell) \mid 0 \le i \le (p-1), 0 \le j \le (p-1)\}.
\end{equation}
In general, for any $0\leq k\leq (p^{2}-1)$, $0 \le i, j \le (p-1)$ and any arbitrary horizontal layer $\ell$ ($0 \le \ell \le (p^{3}-1$)) of $H_{2\text{-D}}$,  we can find the position of 1s by taking into consideration the shifts in \eqref{shifts}. Such positions are given by
\begin{equation}\label{positions}
    S_\ell= \{(ip + \bmod(\ell + a(i, j, k), p), jp + \bmod(\ell + b(i, j, k), p)\}. 
\end{equation}
It is easy to see that $k = \lfloor \frac{\ell}{p} \rfloor$, where $0 \le \ell \le (p^{3}-1$). Thus, after substituting the values of $a(i, j, k)$ and $b(i, j, k)$ from~\eqref{shifts} in~\eqref{positions}, we have
 \begin{equation}\label{generalpositions}
\begin{aligned}
S_\ell
= \Bigg\{\Big(
&\, ip + \bmod\Big(
\ell + \eta\Big(\left\lfloor \frac{\ell}{p} \right\rfloor\Big)
+ \left\lfloor \frac{\ell}{p^2} \right\rfloor (\phi(i)+\psi(j)),\, p
\Big), \\
&\, jp + \bmod\Big(
\ell + \left\lfloor \frac{\ell}{p^2} \right\rfloor \phi(i),\, p
\Big)
\Big)
\Bigg\}.
\end{aligned}
\end{equation}
Putting $\ell = \ell_1$ and $\ell = \ell_2$ in~\eqref{generalpositions}, we get $S_{\ell_1}$ and $S_{\ell_2}$, which are the positions of 1s in layers $\ell_1$ and $\ell_2$. Equating $S_{\ell_1}$ and $S_{\ell_2}$ gives the following two equations:
\begin{align*}
&ip + \bmod \left( \ell_{1} +\eta\left( \left\lfloor \frac{\ell_{1}}{p} \right\rfloor \right)+ \left\lfloor \frac{\ell_{1}}{p^2} \right\rfloor (\phi(i) + \psi(j)), p \right)  \nonumber \\
&=ip + \bmod \left( \ell_{2} +\eta\left( \left\lfloor \frac{\ell_{2}}{p} \right\rfloor\right) + \left\lfloor \frac{\ell_{2}}{p^2} \right\rfloor (\phi(i) + \psi(j)), p \right);  
\end{align*}
{\small\begin{equation*}
   jp + \bmod \left( \ell_{1} + \left\lfloor \frac{\ell_{1}}{p^2} \right\rfloor \phi(i), p \right)= jp + \bmod \left( \ell_{2} + \left\lfloor \frac{\ell_{2}}{p^2} \right\rfloor \phi(i), p \right). 
\end{equation*}}
Simplifying the above equations, we have
\begin{align}
  & \bmod \left( x_{1} + x_{2} + x_{3}(\phi(i)+\psi(j)), p\right) = 0,\label{final_position_eqution1}\\
  & \bmod \left( x_{1} + x_{3}\phi(i), p \right) = 0,\label{final_position_eqution2} 
\end{align}
where 
$x_1 = \ell_1 - \ell_2$, $x_2 = \eta\left(\left\lfloor \frac{\ell_1}{p} \right\rfloor\right) - \eta\left( \left\lfloor \frac{\ell_2}{p} \right\rfloor\right)$ and  $x_3 = \left\lfloor \frac{\ell_1}{p^2} \right\rfloor - \left\lfloor \frac{\ell_2}{p^2} \right\rfloor.$ \\

 \textbf{Existence of Solution of Equations \eqref{final_position_eqution1} and \eqref{final_position_eqution2}:\\} Let $\ell_1$ and $\ell_2$ be two layers from different horizontal block-layers, i.e., $x_3\neq 0$. Then, by equation~\eqref{final_position_eqution2}, we have
 \begin{equation*}
     x_{1} + \phi(i) x_{3} = 0~(\bmod{~p})\iff i=\phi^{-1}\left(-\frac{x_{1}}{x_{3}}\pmod{p}\right). 
 \end{equation*}
By substituting the value of $\phi(i)$ in the equation \eqref{final_position_eqution1},  we obtain 
\begin{align*}
&\bmod \left( x_{1} + x_{2} + x_{3}(\phi(i)+\psi(j)), p\right) = 0\\
  &\iff x_{1} + x_{2} + x_{3}(-\frac{x_{1}}{x_{3}}+\psi(j))=0 \pmod{p}\\
  &\iff x_{2} +x_{3}\psi(j)=0 \pmod{p}\\
  &\iff j=\psi^{-1}\left(-\frac{x_{2}}{x_{3}}\pmod{p}\right).
\end{align*}
Consequently, for every pair of layers $(\ell_{1},\ell_{2})$ belonging to different horizontal block-layers, there exists a
pair $(i,j)$ that satisfies equations~\eqref{final_position_eqution1}
and~\eqref{final_position_eqution2}.

\textbf{Uniqueness of Solution of Equations \eqref{final_position_eqution1} and \eqref{final_position_eqution2}:} The uniqueness follows directly from Lemma~\ref{le:4-cyclefree}. Indeed, the existence of more than one solution to \eqref{final_position_eqution1} and \eqref{final_position_eqution2} would imply the presence of a $4$-cycle, which leads to a contradiction.
This completes the proof. \end{proof}

Viewing the layers of $ H_{2\text{-D}}$ as matrices, motivates the following definition. 

\begin{definition}[Rank of the parity-check tensor]
The maximum number of linearly independent layers in a parity-check tensor $ H_{2\text{-D}}$ is referred to as its rank.
\end{definition}
Using this definition, we now state, in the following theorem, the rank of a parity-check tensor $H_{2\text{-D}}$ formed by selecting a subset of horizontal block-layers.

\begin{theorem}\label{th:ranktheorem}
Let \( \widetilde{H}_{2\text{-D}} \) be a parity-check tensor obtained by selecting any $w$ ($1 \leq w \leq p$) horizontal block-layers from $H_{2\text{-D}}$. Then, $\mathrm{gfrank}(\widetilde{H}_{2\text{-D}}) = p^{2} + (w-1)(p^{2}-1).$
\end{theorem}

\begin{proof}
Assume that \(w=1\). Then, it is evident that  \(\langle \ell_{1}, \ell_{2} \rangle = 0\) for \(\ell_{1} \neq \ell_{2}\) whenever \(\ell_{1}\) and \(\ell_{2}\) belong to the same horizontal block-layer. Consequently, the set of mutually orthogonal layers, that is, layers satisfying \(\langle \ell_{1}, \ell_{2} \rangle = 0\), is linearly independent. Therefore, \(\mathrm{gfrank}(\widetilde{H}_{2\text{-D}}) = p^{2}\).

Taking the shifts defined in \eqref{shifts} into account, it is straightforward to verify that the sum of all layers within any horizontal block-layer is the all-ones matrix, which implies a linear dependence among the layers in horizontal block-layers. Now, consider the \(w\) horizontal block-layers of $\widetilde{H}_{2\text{-D}}$, and remove one horizontal layer from every horizontal block-layer except the first one. We delete the last horizontal layer from every horizontal block-layer except the first one and consider the following linear combination of all the remaining layers in the all horizontal-block layers:
  \begin{small}
      \begin{equation} \label{linearcombi}   \sum_{i_{1}=0}^{p^{2}-1}\alpha_{i_{1}}^{(0)}L^{(0)}_{i_{1}}+\sum_{i_{2}=0}^{p^{2}-2}\alpha_{i_{2}}^{(1)}L^{(1)}_{i_{2}}+\cdots+\sum_{i_{w}=0}^{p^{2}-2}\alpha_{i_{w}}^{(w-1)}L^{(w-1)}_{i_{w}}=0,
	 \end{equation}
  \end{small}
where $L_{i}^{(j)}$ stands for the $i$th layer in the $j$th horizontal block-layer and  the coefficients are binary scalars. Consider the inner product defined in \eqref{eq:innerproduct} with $\pmod{2}$. The remainder of the proof follows directly from Lemma~\ref{maintheorem1} and the same line of reasoning employed in the proof of Theorem~1 in~\cite{kumar2024entanglement}.
\end{proof}

Having established the necessary groundwork, we now present 2-D classical QC-LDPC codes along with their parameters. To this end, we convert the parity-check tensor $H_{2\text{-D}}$ into a 2-D parity-check matrix, as illustrated in Fig.~\ref{fig:2-dparitycheckmatrix}. Using the techniques, we state the following theorem.
\begin{theorem}\label{th:2dclassicalwithprime}
  Let $\widetilde{H}_{2\text{-D}}$ be a parity-check tensor obtained by selecting first $w$ ($1 \leq w \leq p$) horizontal block-layers from $H_{2\text{-D}}$. Then there exists a 2-D classical QC-LDPC code with parameters $[p^{4}, (p^{2}-w+1)(p^{2}-1)]_{2}$ whose Tanner graph is free of cycles of length~\(4\). Moreover, the constructed code has the capability to correct \( p \times p \) burst erasures.
\end{theorem}
\begin{proof} Theorem~\ref{th:ranktheorem} establishes the dimension of the proposed code, while in Lemma~\ref{le:4-cyclefree} we have already proved that  the Tanner graph of the code is free of cycles of length~\(4\).

It remains to show that the code corrects \( p \times p \) burst erasures. To show that we consider the Lemma~2 in~\cite{matcha2018two} which guarantees that a 2-D classical LDPC code with parity-check tensor \(H_{2\text{-D}} = [h_{i,j,k}]\) can correct a 2-D burst erasure of size \(s \times t\) if, for every position \((i,j)\), there exists a \(k\)th layer of \(H_{2\text{-D}}\) such that
\begin{equation}\label{eq:conditionforerasure}
h_{x,y,k} =
\begin{cases}
1, & (x,y) = (i,j),\\
0, & (x,y) \in \{(i+u,j+v): \\&-s<u<s,\,-t<v<t\}\setminus\{(i,j)\}.
\end{cases}
\end{equation}
The condition in~\eqref{eq:conditionforerasure} ensures that the $k$th parity-check equation contains exactly one nonzero entry within the erasure region, located at position \((i,j)\). As a result, this parity-check equation involves only the \((i,j)\)th code symbol from the erased block, enabling its unique recovery.

Finally, by examining the layers in the first block-horizontal layer of $H_{2\text{-D}}$ with the shifts defined in~\eqref{shifts}, it follows directly that the proposed construction satisfies the above condition for a $p \times p$ erasure pattern. Hence, the constructed code is capable of correcting any 2-D burst erasures of size $p \times p$.\end{proof}

In the previous subsection, we have only constructed 2-D classical LDPC codes by tiling tensors of size $p\times p\times p$, where $p$ is an odd prime. Now, we are going to state a result, where the size of the tensor is $p\times p\times p$, where $p$ is composite. 
\begin{theorem} Let $c$, $b$ and $h$ be positive integers, where $p \mid h$.
Let the parity-check tensor $H_{2\text{-D}}$ be constructed by stacking
$c \times b \times h$ permutation tensors of size $p \times p \times p$
along the $i$-, $j$- and $k$-directions, respectively with shifts defined in~\eqref{shifts}. If
\begin{equation*}
    (c-1)\left(\frac{h}{p}-1\right) < p
\quad \text{and} \quad
(b-1)\left(\frac{h}{p}-1\right) < p,
\end{equation*}
then the Tanner graph associated with $H_{2\text{-D}}$ contains no cycles of
length~$4$. Moreover, the 2-D classical QC-LDPC code defined by $H_{2\text{-D}}$ is
capable of correcting any $p \times p$ burst erasure.
\end{theorem}
\begin{proof}
 The $4$-cycle-free conditions can be established by following the same line of reasoning as in the proof of Lemma~\ref{le:4-cyclefree}, while the erasure-correction capability has already been proved in Theorem~\ref{th:2dclassicalwithprime}.
\end{proof}
In the next subsection, we focus on the construction of 2-D classical QC-LDPC codes with girth greater than~$6$.
\subsection{2-D Classical QC-LDPC Codes with Girth Greater than~$6$}
It is well-known that the presence of short cycles—particularly cycles of length $4$—in the Tanner graph of a classical code adversely affects decoding performance. Thus far, our focus has been on constructing codes free of $4$-cycles. This naturally motivates the development of classical codes that avoid both $4$- and $6$-cycles. Accordingly, we proceed to construct 2-D classical QC-LDPC codes that are free of cycles of lengths $4$ and $6$. To this end, we first state the generalized Behrend sequence from~\cite{venkataramanappa2025array}.

 For any integers $d \geq 1$, $n \geq 2$, $k \leq nd^2$ and $\delta \geq 2$, define the set $\mathcal{B}(n,k,d,\delta)$ as follows:
\begin{equation}
   \mathcal{B}(n,k,d,\delta)=\{y: y = \sum_{i=1}^{n-1} y_i(\delta d+1)^{i}\}, 
\end{equation}
where the $y_i$'s are integers subject to the following conditions
\begin{equation*}
     0 \leq x_i \leq d, 
    ||y||^2 = k, \text{ where }  ||y|| = \sqrt{\sum_{i=0}^{n-1}y_i^2},
\end{equation*}
where $k$ is fixed for each element in the set $  \mathcal{B}(n,k,d,\delta)$. 

To construct the parity-check tensor $H_{2\text{-D}}$ of size $c\times b\times h$. We appropriately choose the parameters \( n \), \( d \) and $\delta=\left(\frac{h}{p}-1\right)$, where $p\mid h$, for the set \( \mathcal{B}(n,k,d,\delta) \), such that the cardinality of \( \mathcal{B}(n,k,d,\delta) \) is at least $\text{max}\{c,b\}$. We, further, impose the following condition:
\begin{equation}
    p\geq\left(\frac{h}{p}-1\right)\left(\text{max}\{b\in\mathcal{B}(n,k,d,\delta))\}\right)+1.
\end{equation}
We now proceed to construct 2-D classical QC-LDPC codes free of $4$- and $6$-cycles using the following shifts:
\begin{equation}\label{eq:shifts2}
   \left.
\begin{aligned}
a(i, j, k) = \bmod\left(\eta(k) + \left\lfloor \frac{k}{p} \right\rfloor (\phi(i) + \psi(j)), p\right), \\
b(i, j, k)= \bmod\left(\left \lfloor \frac{k}{p} \right\rfloor \phi(i), p\right),
\end{aligned}
\right\} 
\end{equation}
where $\phi:\mathbb{Z}_{c}\to \mathcal{B}(n,k,d,\delta)$ and  $\psi:\mathbb{Z}_{b}\to \mathcal{B}(n,k,d,\delta)$ are injective (one-one) mappings, and the mapping $\eta:\mathbb{Z}_{h}\to\mathbb{Z}_{p}$ is  defined as:
\begin{equation*}
    \eta(k)=\begin{cases}
        \zeta_{1}(k), &\text{ for } 0\leq k\leq (p-1);\\
         \zeta_{2}(\bmod{(k,p)}), &\text{ for } p\leq k\leq (2p-1);\\
           \qquad \vdots &\qquad\vdots\\
           \zeta_{\frac{h}{p}}(\bmod{(k,p)}), &\text{ for } p\leq k\leq ((\frac{h}{p})p-1);
    \end{cases}
\end{equation*}
where $\zeta_{i}:\mathbb{Z}_{p}\to\mathbb{Z}_{p}$ is a bijective mapping for each $i$.

By applying the shifts in \eqref{eq:shifts2}, we construct a 2-D classical QC-LDPC codes, whose Tanner graph is free  from $4$- and $6$-cyles, in the following theorem.
\begin{theorem} Let $H_{2\text{-D}}$ be a parity-check tensor of dimension $c \times b \times h$ along the $i$-, $j$-, and $k$-directions, respectively stacked with shifts specified in~\eqref{eq:shifts2}, and let $\mathcal{C}$ denote the associated 2-D classical QC-LDPC code. Then the Tanner graph of $\mathcal{C}$ contains no cycles of length $4$ or $6$. Moreover, the constructed code is capable of correcting any $p \times p$ burst erasure.
\end{theorem}
 \begin{proof} First, we prove that the Tanner graph of the code does not have $4$-cycles. To, prove this, consider arbitrary horizontal block-layers $k_{0}$ and $k_{1}$, where $0\leq k_{0}, k_{1}\leq (\frac{h}{p}-1)$. To check the existence of $4$-cycle, consider the following arbitrary closed path of length $4$ of tensors indices:
     \begin{equation*}
(i_{0},j_{0},k_{0}),(i_{1},j_{1},k_{0});
(i_{1},j_{1},k_{1}),(i_{0},j_{0},k_{1}),
\end{equation*}
where $i_s \in \{0,1,2,\ldots,(c-1)\}$, $ j_t \in \{0,1,2,\ldots,(b-1)\}$ for all $s,t \in \{0,1\}$.
Then, by Theorem~\ref{thm:2D-2k-cycle}, a $4$-cycle exists along the closed path
$(i_{0}, j_{0}, k_{0}), (i_{1}, j_{1}, k_{0}), (i_{1}, j_{1}, k_{1}), (i_{0}, j_{0}, k_{1})$
if and only if there exists a set $\{\bar{k}_{0}, \bar{k}_{1}\}$ of $k$-indices, where each
$\bar{k}_{i}$ is the $k$-coordinate of a tensor in the horizontal block layer $k_{i}$, satisfying the conditions required to connect the corresponding cycle segments, and the following equations are satisfied simultaneously:
\begin{align*}
\sum_{m=0}^{1}
\Bigl(
a(i_{m},j_{m},\Bar{k}_{m})
-
a(i_{m+1},j_{m+1},\Bar{k}_{m})
\Bigr)
&= 0 \pmod p,\\
\sum_{m=0}^{1}
\Bigl(
b(i_{m},j_{m},\Bar{k}_{m})
-
b(i_{m+1},j_{m+1},\Bar{k}_{m})
\Bigr)
&=0 \pmod p, 
\end{align*}
where $(i_2,j_2)=(i_0,j_0)$. 
Now, by substituting the shifts defined in \eqref{eq:shifts2} into the above equations, a $4$-cycle exists in the Tanner graph of $H_{2\text{-D}}$ if and only if the following equations are satisfied:
\begin{equation}\label{eq:4cyclein46}
   \alpha((\phi(i_0)-\phi(i_{1}))+( \psi(j_0)-\psi(j_{1})))=0\pmod{p},
\end{equation}
\begin{equation}\label{eq:4cyclein46phi}
  \alpha ((\phi(i_0)-\phi(i_{1}))=0\pmod{p},
\end{equation}
where $\alpha= \left(\left\lfloor \frac{\Bar{k}_{1}}{p}\right\rfloor-\left\lfloor \frac{\Bar{k}_{0}}{p}\right\rfloor\right)\neq0\pmod{p}$ and less than or equal to $\left(\frac{h}{p}-1\right)$. Consider the equation~\eqref{eq:4cyclein46phi}, by Lemma 1 in \cite{venkataramanappa2025array}, we have $\phi(i_0)=\phi(i_{1})$ which implies that $i_{0}=i_{1}$. By putting $i_{0}=i_{1}$ in equation~\eqref{eq:4cyclein46}, we have
\begin{equation}
   \alpha( \psi(j_0)-\psi(j_{1}))=0.  
\end{equation}
Again by Lemma 1 in \cite{venkataramanappa2025array}, we have $\psi(j_0)=\psi(j_{1})$ which implies that $j_{0}=j_{1}$, which is a contradiction since $(i_{0},j_{0})\neq(i_{1},j_{1})$. Consequently, we can say that the Tanner graph of $H_{2\text{-D}}$ is free from $4$-cycles.

Next, we proceed to prove that the Tanner graph of $H_{2\text{-D}}$ is free from $6$-cycles. To prove that consider arbitrary horizontal block-layers $k_{0}$, $k_{1}$ and $k_{2}$, where $0\leq k_{0}, k_{1},k_{2}\leq (\frac{h}{p}-1)$. To check the existence of $6$-cycle, consider the  following arbitrary closed path of length $6$ of tensors indices: $
(i_{0},j_{0},k_{0}),(i_{1},j_{1},k_{0});$
$(i_{1},j_{1},k_{1}),(i_{2},j_{2},k_{1});$ $
(i_{2},j_{2},k_{2}),$ $(i_{0},j_{0},k_{2}),$
where $i_s \in \{0,1,2,\ldots,(c-1)\}$ and $j_t \in \{0,1,2,\ldots,(b-1)\}$ for all $s,t \in \{0,1\}$.

 Then, by Theorem~\ref{thm:2D-2k-cycle}, a $6$-cycle exists along the closed path
$(i_{0},j_{0},k_{0}),(i_{1},j_{1},k_{0});$
$(i_{1},j_{1},k_{1}),(i_{2},j_{2},k_{1});$ $
(i_{2},j_{2},k_{2}),$ $(i_{0},j_{0},k_{2}),$
if and only if there exists a set $\{\bar{k}_{0}, \bar{k}_{1}, \bar{k}_{2}\}$ of $k$-indices, where each
$\bar{k}_{i}$ is the $k$-coordinate of a tensor in the horizontal block layer $k_{i}$,
satisfying the conditions required to connect the corresponding cycle segments, and the following equations are simultaneously satisfied: 
\begin{align*}
\sum_{m=0}^{2}
\Bigl(
a(i_{m},j_{m},\Bar{k}_{m})
-
a(i_{m+1},j_{m+1},\Bar{k}_{m})
\Bigr)
&= 0 \pmod p,\\
\sum_{m=0}^{2}
\Bigl(
b(i_{m},j_{m},\Bar{k}_{m})
-
b(i_{m+1},j_{m+1},\Bar{k}_{m})
\Bigr)
&=0 \pmod p, 
\end{align*}
where $(i_3,j_3)=(i_0,j_0)$. 
Now, by substituting the shifts defined in \eqref{eq:shifts2} into the above equations, a $6$-cycle exists in the Tanner graph of $H_{2\text{-D}}$ if and only if the following equations are satisfied under modulo $p$:
\begin{equation}\label{eq:6cyclecondition1}
     \gamma(\phi(i_0)+\psi(j_{0}))= \beta(\phi(i_2)+\psi(j_{2}))+
     \alpha(\phi(i_1)+\psi(j_{1})),
\end{equation}
\begin{equation}\label{eq:6cyclecondition2}
       \gamma \phi(i_0)= \beta\phi(i_2)+ \alpha\phi(i_1),
\end{equation}
where $\alpha=\left(\left\lfloor \frac{\Bar{k}_{1}}{p}\right\rfloor-\left\lfloor \frac{\Bar{k}_{0}}{p}\right\rfloor\right)$, $\beta=\left(\left\lfloor \frac{\Bar{k}_{2}}{p}\right\rfloor-\left\lfloor \frac{\Bar{k}_{1}}{p}\right\rfloor\right)$ and $\gamma=\left(\left\lfloor \frac{\Bar{k}_{2}}{p}\right\rfloor-\left\lfloor \frac{\Bar{k}_{0}}{p}\right\rfloor\right)$. It is easy to see that $\gamma=\alpha+\beta$. Now, consider the equation~\eqref{eq:6cyclecondition1}
\begin{equation*}
       \gamma \phi(i_0)= \beta\phi(i_2)+ \alpha\phi(i_1),
\end{equation*} 
if any two or three $\phi(i)$'s are equal then use Lemma 1 from \cite{venkataramanappa2025array} otherwise use Lemma 2 from
\cite{venkataramanappa2025array}. Applying Lemma~1 and Lemma~2 from~\cite{venkataramanappa2025array} to the above equation yields \( i_{0} = i_{1} = i_{2} \). Substituting \( i_{0} = i_{1} = i_{2} \) into the equation~\eqref{eq:6cyclecondition1} gives \( j_{0} = j_{1} = j_{2} \), which contradicts the definition of a closed path. This completes the proof.
\end{proof}
\section{2-D Entanglement-assisted QLDPC codes}\label{sec:quantumconstruction}
In this section, we construct two classes of 2-D entanglement-assisted quantum LDPC codes by emplyoing proposed 2-D classical QC-LDPC codes. The first class is constructed from two distinct 2-D classical QC-LDPC codes, and the second class is obtained by employing a single 2-D classical QC-LDPC code. Notably, the construction of the first class requires only a single ebit, and the unassisted portion of the overall Tanner graph is free of cycles of length $4$. To understand the meaning of the unassisted portion of the overall graph, the reader may refer to \cite{kumar2024entanglement}.

Now, consider the parity-check tensor \(H_{2\text{-D}}\) of dimension $p\times p\times p^{2}$, where stacked permutation tensors are of size $p\times p\times p$, constructed using the shifts defined in \eqref{shifts}. We partition \(H_{2\text{-D}}\) into two sub parity-check tensors, denoted by \(H^{(1)}_{2\text{-D}}\) and $H^{(2)}_{2\text{-D}}$. The first sub parity-check tensor consists of first $w_{1}$ horizontal block-layers, where $1 \leq w_{1} \leq (p-1)$, and the second consists of $w_{2}$ horizontal block-layers, where $1 \leq w_{2} \leq (p-1)$, such that $2 \leq w_{1}+w_{2} \leq p$. Moreover, $H^{(1)}_{2\text{-D}}$ and $H^{(2)}_{2\text{-D}}$ do not share any common horizontal block-layer. We then obtain the following theorems based on~\cite[Corollary~1]{wilde2008optimal}.
\begin{theorem}
    Let $\mathcal{C}_{1}$ and $\mathcal{C}_{2}$ be 2-D classical QC-LDPC codes with  $H^{(1)}_{2\text{-D}}$ and $H^{(2)}_{2\text{-D}}$ as their parity-check tensors. Then there exists a 2-D EA quantum LDPC code with the parameters $[[p^{4}, p^{4}-2p^{2}-(p^{2}-1)(w_{1}+w_{2}-2)+1;1]]_{2}.$ Moreover, its unassisted portion of overall Tanner graph is of girth $>4$.
\end{theorem}
\begin{proof} The proof follows directly from Lemma~\ref{maintheorem1} and Theorem~\ref{th:ranktheorem}.
\end{proof}
\begin{theorem}
    Let $\mathcal{C}$ be a 2-D classical QC-LDPC code and $\widetilde{H}_{2\text{-D}}$ be its parity-check tensor having distinct $w$ number of horizontal block-layers. Then there exists a 2-D EA quantum LDPC code, with the parameters $[[p^{4},(p^{2}-1)(p^{2}-w+1);p^{2}+(w-1)(p^{2}-1)]]_{2}$. In addition, by construction the Tanner graph of $\mathcal{C}$ is free from cyles of length $4$.
\end{theorem}
\begin{proof}
    Let $\widetilde{H}_{1\text{-D}}$ denote the 2-D representation of $H^{(1)}_{2\text{-D}}$. By adapting the arguments used in the proof of Theorem~4 in~\cite{kumar2024entanglement} to the present setting, the desired result follows.
\end{proof}
In this section, we have thus far constructed EA quantum LDPC codes by employing the proposed 2-D classical QC-LDPC codes with the prime parameter $p$. In a similar manner, 2-D EA quantum codes can also be constructed from the other 2-D classical QC-LDPC codes developed earlier. However, in those cases, analytically determining the code rate and the required number of entangled bits becomes difficult.
\section{Conclusions}\label{sec:conclusion}
First, we derive a general condition for the existence of $2g$-cycle in 2-D classical QC-LDPC code. We then construct a family of 2-D classical QC-LDPC codes with 2-D burst-erasure correction capability and determine their exact code rates. In addition, we show that their Tanner graphs are free of cycles of length~$4$. Several other families are also constructed with higher girth and the same erasure-correction capability. We employ the constructed 2-D classical QC-LDPC codes to develop two families of entanglement-assisted quantum LDPC codes. In the first family, two distinct 2-D classical QC-LDPC codes are used to ensure that the unassisted portion of the overall Tanner graph is free of $4$-cycles, while requiring only a \textit{single ebit}. The second family is constructed using a single 2-D classical QC-LDPC code, where the underlying classical code itself is free of $4$-cycles.
\section*{Acknowledgment}
P. Kumar is supported by the Institute Post-Doctoral Fellowship from the Indian Institute of Science, Bengaluru, India. S. S. Garani is supported by a grant from the Anusandhan National Research Foundation, Government of India, under Grant No. SERB/F/3132/2023–2024.
\bibliography{mybibliography}
\end{document}